\documentclass[a4paper,12pt]{article}

\usepackage{amsmath}
\usepackage{amssymb}
\usepackage{graphicx}
\usepackage[OT4]{fontenc}
\usepackage[utf8]{inputenc}
\usepackage{tikz}
\usepackage{url}
\usepackage[numbers]{natbib}
\usepackage{bm}
\usepackage[unicode,debug,colorlinks=true,linkcolor=black,citecolor=black,urlcolor=black,pdfusetitle=true,breaklinks=true]{hyperref}
\allowdisplaybreaks

\usepackage{setspace}

\newcommand{\abs}[1]{\left| #1 \right|}

\newcommand{\okra}[1]{\left( #1 \right)}
\newcommand{\kwad}[1]{\left[ #1 \right]}
\newcommand{\klam}[1]{\left\{ #1 \right\}}

\newcommand{\floor}[1]{\left\lfloor #1 \right\rfloor}

\DeclareMathOperator*{\hilberg}{hilb}

\DeclareMathOperator{\esssup}{ess\, sup}
\DeclareMathOperator{\card}{card}

\DeclareMathOperator{\cov}{Cov}
\DeclareMathOperator{\var}{Var}
\DeclareMathOperator{\mean}{\mathbf{E}}
\DeclareMathOperator{\median}{\mathbf{M}}

\newcommand{\boole}[1]{{\bf 1}{\klam{#1}}}

\newtheorem{definition}{Definition}
\newtheorem{example}{Example}

\newtheorem{proposition}{Proposition}
\newtheorem{theorem}{Theorem}
\newtheorem{lemma}{Lemma}
\newenvironment*{proof}{\begin{trivlist}\item[]
\noindent\textbf{Proof:}}{$\Box$\par\end{trivlist}}
\newenvironment*{proof*}[1]{\begin{trivlist}\item[]
\noindent\textbf{Proof of #1:}}{$\Box$\par\end{trivlist}}

\author{{\L}ukasz D\k{e}bowski\thanks{
    {\L}. D\k{e}bowski is with
    the Institute of Computer Science, Polish Academy of Sciences, 
    ul. Jana Kazimierza 5, 01-248 Warszawa, Poland 
    (e-mail: ldebowsk@ipipan.waw.pl).
  }
}

\title{Repetition and recurrence times: \\ Dual statements and
  summable mixing rates} \date{}

\begin{document}

\begin{titlepage}
\maketitle

\begin{abstract}
  By an analogy to the duality between the recurrence time and
  the longest match length, we introduce a quantity dual to the
  maximal repetition length, which we call the repetition time.
  Extending prior results, we sandwich the repetition time in terms of
  unconditional and conditional min-entropies. The upper bound holds
  if the mixing rate $\phi(n)$ is summable, whereas the lower bound
  only assumes stationarity.  Our reasoning makes a repeated use of
  dualities between so-called times and so-called counts that
  generalize the duality of the recurrence time and the longest match
  length.  We also discuss the analogy of these results with the
  Wyner-Ziv/Ornstein-Weiss theorem, which sandwiches the recurrence
  time in terms of Shannon entropies.
%
  \\[1em]
  \\[1em]
\end{abstract}


\end{titlepage}

\section{Introduction}
\label{secIntroduction}

Using appropriate concepts is often a gate to a simple theory. A
modest instance of this rule will be exhibited in this article. The
paper concerns stochastic bounds for the maximal repetition length and
its newly introduced dual---called the repetition time.  Our
development will be analogous to the known duality between the
recurrence time and the longest match length \cite{WynerZiv89}.
Whereas the repetition time is defined in this paper, there is a long
history of investigation of the recurrence time
\cite{Kac47,WynerZiv89,OrnsteinWeiss93,Kontoyiannis98,Ko12}, the
longest match length \cite{WynerZiv89,OrnsteinWeiss93,Szpankowski93,
  KontoyiannisOther98,GaoKontoyiannisBienenstock08}, and the maximal
repetition length
\cite{DeLuca99,KolpakovKucherov99a,KolpakovKucherov99,CrochemoreIlie08,
  ErdosRenyi70,ArratiaWaterman89,Shields92b,Shields97,
  Szpankowski93,Szpankowski93a,KontoyiannisSuhov94,
  Debowski11b,Debowski17,Debowski18b}. These concepts appear in
various contexts in ergodic theory, information theory, theoretical
computer science, and quantitative linguistics, as we have indicated
by the references and we will detail further.

To be concrete, let $(X_i)_{i\in\mathbb{Z}}$ be a stochastic process
over a countable alphabet. We denote blocks
$X_j^k:=(X_j,X_{j+1},\ldots,X_k)$. The well-researched recurrence time
$R^{(1)}_k$ is the first position in the future on which a copy of
string $X_1^k$ reappears. By contrast, we define the repetition time
$R^{(2)}_k$ as the first position in the future on which a copy of any
previous string $X_{j+1}^{j+k}$ occurs. In formulas,
\begin{align}
  R^{(1)}_k&:=\inf\klam{\infty}\cup\klam{i\ge 1: X_{i+1}^{i+k}=X_1^k},
  \\
  R^{(2)}_k&:=\inf\klam{\infty}\cup\klam{i\ge 1:
  X_{i+1}^{i+k}=X_{j+1}^{j+k}\text{ for some } 0\le j<i}.
\end{align}
If $(X_i)_{i\in\mathbb{Z}}$ is a process over a $D$-ary alphabet then
$R^{(2)}_k\le D^k$. In any case, $R^{(2)}_k\le R^{(1)}_k$. Precisely,
we have
$R^{(2)}_k=\inf_{i\ge 0}\okra{i+R^{(1)}_k\circ T^i}$
for the shift automorphism $X_i\circ T=X_{i+1}$, since observing a
repetition is equivalent to observing a recurrence of any string.

In parallel, we can discuss variables that are dual to random
variables $R^{(\gamma)}_k$. The longest match length $L^{(1)}_n$
and the maximal repetition length $L^{(2)}_n$ are
\begin{align}
  L^{(1)}_n&:=\max\klam{0}\cup\klam{k\ge 1: X_{i+1}^{i+k}=X_1^k
  \text{ for some } 0< i\le n-k},
  \\
  L^{(2)}_n&:=\max\klam{0}\cup\klam{k\ge 1: X_{i+1}^{i+k}=X_{j+1}^{j+k}
  \text{ for some } 0\le j< i\le n-k}.
\end{align}
Both lengths $L^{(\gamma)}_n$ are functions of prefix $X_1^n$ and we have
$L^{(1)}_n\ge L^{(2)}_n\le n$. For $\gamma\in\klam{1,2}$, we have the
dualities
\begin{align}
  R^{(\gamma)}_k&=\inf\klam{\infty}\cup\klam{n\ge 1: L^{(\gamma)}_{n+k}\ge k},
  &
  L^{(\gamma)}_n&=\max\klam{0}\cup\klam{k\ge 1: R^{(\gamma)}_k\le n-k}.
\end{align}
Hence $R^{(\gamma)}_k>n\iff L^{(\gamma)}_{n+k}<k$
\cite{WynerZiv89}. Thus the sequence of the dual random variables
carries the same information as the sequence of the primal ones. Both
$R^{(\gamma)}_k$ and $L^{(\gamma)}_n$ are non-decreasing functions of
$k$ and $n$, respectively.

Let us describe some related received theoretical results in the
stochastic setting. We begin with the recurrence time and the longest
match length.  Let $\log x$ denote the natural logarithm of $x$.
Consider a discrete random variable $X$. The Shannon entropy is
$H_1(X):=\mean\kwad{-\log P(X)}$. Let $(X_i)_{i\in\mathbb{Z}}$ be a
stationary ergodic process over a finite alphabet. The celebrated
Wyner-Ziv/Ornstein-Weiss (WZ/OW) theorem
\cite{WynerZiv89,OrnsteinWeiss93} asserts
\begin{align}
  \label{WZOW}
  \lim_{k\to\infty} \frac{\log n}{L^{(1)}_n}
  =\lim_{k\to\infty} \frac{1}{k}\log R^{(1)}_k
  =h_1:=\lim_{k\to\infty} \frac{H_1(X_1^k)}{k}
  \text{ a.s.}
\end{align}
Thus the rates of the recurrence time and the longest match length are
essentially dictated by the rate of the Shannon entropy.

Let us proceed to the maximal repetition length and the repetition
time.  For memoryless sources, the maximal repetition length
$L^{(2)}_n$ grows at a similar rate as the longest match length
$L^{(1)}_n$ \cite{ErdosRenyi70,ArratiaWaterman89}. Probably for this
reason, the problem that $L^{(2)}_n$ can grow much faster than
$L^{(1)}_n$ for other stationary processes was first recognized as
late as by Shields \cite{Shields92b,Shields97}. Define the collision
entropy as $H_2(X):=-\log\mean P(X)$, where $H_2(X)\le H_1(X)$ by the
Jensen inequality. We also denote the $\psi$-mixing rate
\begin{align}
  \psi(n):=
  \sup_{A\in\mathcal{X}_{-\infty}^0,B\in\mathcal{X}_n^{\infty}}
  \abs{\frac{P(A\cap B)}{P(A)P(B)}-1},
\end{align}
where $\mathcal{X}_j^k$ is the $\sigma$-field generated by block
$X_j^k$.  A technically complicated result by Szpankowski
\cite[Theorem 2(ii)]{Szpankowski93}, \cite{Szpankowski93a}, involving
the height of suffix trees, can be rewritten as
\begin{align}
  \label{Szpankowski}
  \lim_{k\to\infty} \frac{\log n}{L^{(2)}_n}
  =\lim_{k\to\infty} \frac{1}{k}\log R^{(2)}_k
  =h_2:=\lim_{k\to\infty} \frac{H_2(X_1^k)}{k}
  \text{ a.s.}
\end{align}
provided we have condition
\begin{align}
  \label{PsiSquareSummable}
  \sum_{n=1}^\infty \kwad{\psi(n)}^2<\infty.
\end{align}
Mind that Szpankowski \cite{Szpankowski93} denoted the mixing rate
$\psi(n)$ as $\alpha(n)$, whereas we use the nomenclature by Bradley
\cite{Bradley05}.  Whereas the Shannon entropy rate $h_1$ exists for
any stationary process, the collision entropy rate $h_2$ exists under
condition (\ref{PsiSquareSummable}) in particular. The collision
entropy rate can be substantially smaller than the Shannon entropy
rate.

Condition (\ref{PsiSquareSummable}) is satisfied by any finite-state
aperiodic Markov chain. In this case, mixing rate $\psi(n)$ decays
exponentially \cite[Theorems 3.1]{Bradley05}. However, condition
(\ref{PsiSquareSummable}) does not cover all important cases---such as
some stationary sources that do not satisfy the finite-energy
condition \cite{Shields92b,Shields97}, do not meet the Doeblin
condition \cite[Theorem 11.7]{Debowski21}, or are not diluted by
random noise \cite{Shields97}, \cite[Theorem 11.6]{Debowski21}. Here
we admit that we became interested in characterizing the maximal
repetition length because of our long-standing interests in
statistical language modeling \cite{Debowski21}.  In fact, for natural
language, experiments suggest that the Shannon entropy rate is
positive, $h_1>0$, \cite{Shannon51,CoverKing78} but the Shannon mutual
information between two blocks of an increasing length grows like a
power-law, called Hilberg's law or the neural scaling law
\cite{Hilberg90, TakahiraOther16, KaplanOther20, Debowski21}, which
is incompatible with condition $\psi(1)<\infty$. Moreover, as
discovered in \cite{Debowski12b,Debowski15f} for novels in English,
German, and French, we have a power-law logarithmic growth of the
maximal repetition length
\begin{align}
  L^{(2)}_n\propto \log^\alpha n,\quad \alpha\approx 3.
\end{align}
This translates to the stretched exponential growth of the repetition
time
$\log R^{(2)}_k\propto k^{1/\alpha}$.

The goal of this paper is to approach theoretical bounds for
repetition and recurrence times by strengthening and beautifying our
earlier theoretical accounts and cumbersome statements
\cite{Debowski15f,Debowski18b,Debowski21}. We will do it by using the
minimal times $R^{(\gamma)}_k$ rather than the maximal lengths
$L^{(\gamma)}_n$.  A few explanations regarding our framework are in
place.

First, the Hilberg exponent of a sequence $(a_k)_{k\in\mathbb{N}}$ is
defined as
\begin{align}
  \label{Exponent}
  \hilberg_{k\rightarrow\infty} a_k
  :=\max\klam{0,\limsup_{k\rightarrow\infty} \frac{\log
  a_k}{\log k}}
  .
\end{align}
It measures the rate of asymptotic power-law growth,
$\hilberg_{k\rightarrow\infty} k^\beta=\max\klam{0,\beta}$.
Properties of the Hilberg exponent are summarized in \cite[Section
8.1]{Debowski21}. We write succinctly
\begin{align}
  a_k\lesssim b_k&\iff
  \hilberg_{k\rightarrow\infty} a_k\le
  \hilberg_{k\rightarrow\infty} b_k,
  &
  a_k\sim b_k&\iff
  \hilberg_{k\rightarrow\infty} a_k=
  \hilberg_{k\rightarrow\infty} b_k.
\end{align}
The Hilberg exponent is useful for bounding the power-law growth of
Shannon mutual information, i.e., Hilberg's law \cite[Chapter
8]{Debowski21}. Hence comes its name.

Second, for a parameter $\gamma\in(0,1)\cup(1,\infty)$, the
R\'enyi-Arimoto conditional entropy of a discrete random variable $X$
given a variable $Y$ is defined as
\begin{align}
  \label{RenyiArimoto}
  H_\gamma(X|Y):=
  -\frac{\gamma}{\gamma-1}\log
  \mean\kwad{\sum_x P(X=x|Y)^{\gamma}}^{1/\gamma}
  ,
\end{align}
whereas for $\gamma\in\klam{0,1,\infty}$, we define
$H_\gamma(X|Y):=\lim_{\delta\rightarrow\gamma} H_\delta(X|Y)$.  We
also put $H_\gamma(X):=H_\gamma(X|1)$. The special cases are called:
\begin{itemize}
\item the Hartley entropy
$H_0(X|Y)= \log \esssup\card\klam{x: P(X=x|Y)>0}$,
\item  the Shannon
entropy $H_1(X|Y)=\mean\kwad{-\log P(X|Y)}$, 
\item the collision entropy
$H_2(X|Y)=-2\log \mean\kwad{\sum_x P(X=x|Y)^2}^{1/2}$,
\item the min-entropy $H_\infty(X|Y)=-\log \mean\max_{x} P(X=x|Y)$.
\end{itemize}
A survey of R\'enyi-Arimoto entropies can be found in
\cite{FehrBerens14}, which was repeated also in \cite[Section
9.1]{Debowski21}. We have $H_\gamma(X)\ge H_\delta(X)$ for
$\gamma<\delta$, whereas
$H_\gamma(X)\le \frac{\gamma}{\gamma-1}H_\infty(X)$ for $\gamma>1$.
Thus all R\'enyi-Arimoto entropies with $\gamma>1$ are of the same
magnitude.  Moreover, we have a generalized chain rule
\begin{align}
  \label{Chain}
  H_\gamma(X|Y,Z)\le H_\gamma(X|Y)\le H_\gamma(U,X|Y)\le
  H_0(U|Y)+H_\gamma(X|Y,U).
\end{align}

Let us proceed to the exposition of the results of this paper.  In the
body of the paper, we will prove the following Theorems
\ref{theoMainOne}--\ref{theoMainThree}, which bound the recurrence and
repetition times comprehensively. The main novelty can be attributed
only to Theorem \ref{theoMainTwo}.
\begin{theorem}
  \label{theoMainOne}
  Consider a stationary ergodic process $(X_i)_{i\in\mathbb{Z}}$ over
  a $D$-ary alphabet. Define $N_k:=\max\klam{N: N\log N\le
    D^k}$. We have a sandwich bound
  \begin{align}
    \label{MainOne}
    \log P(X_1^k|X_{k+1}^{N_k})\lesssim
    \log R^{(1)}_k\lesssim
    -\log P(X_1^k)\lesssim
    H_1(X_1^k)
    \text{ a.s.},
  \end{align}
  where the right-most inequality $\lesssim$ becomes the equivalence
  $\sim$ for
  \begin{align}
    \limsup_{k\to\infty} \frac{\sqrt{\var\log P(X_1^k)}}{H_1(X_1^k)}<1.
  \end{align}
\end{theorem}
\noindent
\emph{Remark:} The WZ/OW theorem (\ref{WZOW}) and the SMB theorem
\cite{Shannon48, Breiman57} imply the sandwich bound (\ref{MainOne})
with all inequalities $\lesssim$ becoming equivalences $\sim$ for
$h_1>0$. We display (\ref{MainOne}) as a separate fact since it also
holds for $h_1=0$. Theorem \ref{theoMainOne} applies the same sandwich
bound by Kontoyiannis \cite{Kontoyiannis98} as the simple proof of the
WZ/OW theorem (\ref{WZOW}) presented in \cite[Theorem
9.12]{Debowski21}.  We also recall that
$h_1=\mean\kwad{-\log P(X_1|X_2^\infty)}$ for a stationary
process. Consequently if $h_1=0$ then $P(X_1^k|X_{k+1}^\infty)=1$
almost surely. The quantity $\var\log P(X_1^k)$ is called varentropy
\cite{KontoyiannisVerdu13}.

For the next theorem, we will apply a concept of short memory that is
somewhat different from that in \cite{Szpankowski93}.  This condition
combines the $\phi$-mixing rather than the $\psi$-mixing condition
\cite{Bradley05} with various ideas of summability of correlations,
usually termed short memory in time series analysis \cite{Beran94}.
\begin{definition}[short memory] For a stationary process
  $(X_i)_{i\in\mathbb{Z}}$, we define counts
\begin{align}
  F_n(x_1^k)&:=\sum_{i=0}^{n-1}\boole{X_{i+1}^{i+k}=x_1^k}.
\end{align}
We say that short memory holds when $\log \gamma_k\sim 0$,
where
\begin{align}
  \label{GammaDelta}
  \gamma_k&:=\sup_{n\in\mathbb{N}}\max_{x_1^k}
            \frac{\var F_n(x_1^k)}{\mean F_n(x_1^k)}
            .
\end{align}
\end{definition}
\noindent
\emph{Remark:} Short memory is a different condition than
(\ref{PsiSquareSummable}). It can be easily seen that short memory is
implied by the uniform absolute summability of correlations,
which---using a more abstract language---can be implied by the
summability of the $\phi$-mixing rate
\begin{align}
  \phi(n):=
  \sup_{A\in\mathcal{X}_{-\infty}^0,B\in\mathcal{X}_n^{\infty}}\abs{P(B|A)-P(B)},
\end{align}
defined as in \cite{Bradley05}. Namely, assuming stationarity, we may
bound
\begin{align}
  \gamma_k
  &\le
    \sup_{n\in\mathbb{N}}\max_{x_1^k}
    \frac{\sum_{i=0}^{n-1}\sum_{j=0}^{n-1}
    \abs{\cov(\boole{X_{1}^{k}=x_1^k};\boole{X_{j-i+1}^{j-i+k}=x_1^k})}}{
    \sum_{i=0}^{n-1}\mean\boole{X_{1}^{k}=x_1^k}}
    \nonumber\\
  &\le
    \max_{x_1^k}
    \frac{(2k-1)\var\boole{X_{1}^{k}=x_1^k}
    +
    2\sum_{n=1}^\infty
    \abs{\cov(\boole{X_{1}^{k}=x_1^k};\boole{X_{n+k+1}^{n+2k}=x_1^k})}}{
    \mean\boole{X_{1}^{k}=x_1^k}}
    \nonumber\\
  &\le
    2k-1+2\sum_{n=1}^\infty \phi(n)
    .
  \label{Summable}
\end{align}
Thus short memory is implied by condition
\begin{align}
  \label{PhiSummable}
  \sum_{n=1}^\infty \phi(n)<\infty.
\end{align}
Condition (\ref{PhiSummable}) cannot be directly related to condition
(\ref{PsiSquareSummable}). We have $\phi(n)\le \psi(n)/2$
\cite[Eq. (1.11)]{Bradley05} but $\psi(n)\ge \kwad{\psi(n)}^2$ for all
but finitely many $n$ if $\lim_{n\to\infty} \psi(n)=0$.  Some basic
examples of a process with short memory 
are not only finite-state aperiodic Markov chains or but also
countable Markov chains with $\phi(n)<1/2$ for some $n\ge 1$. In both
cases, mixing rate $\phi(n)$ decays exponentially \cite[Theorems 3.1
and 3.3]{Bradley05}.

Consequently, we may demonstrate this fact.
\begin{theorem}
  \label{theoMainTwo}
  Consider a stationary process $(X_i)_{i\in\mathbb{Z}}$ over a finite
  alphabet. We have
  \begin{align}
    \label{MainTwoOne}
    \text{short memory } 
    &\implies
    \log R^{(2)}_k\lesssim
    H_\infty(X_1^k)
      \text{ a.s.}
  \end{align}
\end{theorem}
\noindent
\emph{Remark 1:} In \cite{Debowski18b,Debowski21}, a trivial bound in
terms of Shannon entropy $H_1(X_1^k)$ stemming from inequality
$L^{(2)}_n\ge L^{(1)}_n$ was discussed instead of the upper bound
(\ref{MainTwoOne}). In a previous preprint version of this article, we
claimed erroneously that $\log R^{(2)}_k$ can be generally upper
bounded by $H_\infty(X_1^k)$. The phenomenon of burstiness, described
in Section \ref{secRepetition} likely excludes such a possibility.
\\
\emph{Remark 2:} The class of short memory processes seems too
restrictive for linguistic applications due to Taylor's law. Taylor's
law is an empirical power law that captures scaling of variance in
natural language
\cite{GerlachAltmann14,KobayashiTanakaIshii2018,TanakaIshii21} and it
may imply insummable correlations. This hinders the application of
bound (\ref{MainTwoOne}) to natural language. Thus we are not fully
satisfied with Theorem \ref{theoMainTwo}.

We notice also a general lower bound for all stationary ergodic
processes. This theorem somewhat strengthens \cite[Theorem
8]{Debowski18b}.
\begin{theorem}
  \label{theoMainThree} 
  Consider a stationary process $(X_i)_{i\in\mathbb{Z}}$ over a
  countable alphabet. Define the context length
  $I_k:=\min\klam{i:\log i\ge H_\infty(X_1^k|X_{k+1}^{k+i})}$.
  We have a lower bound
  \begin{align}
    \label{MainThree}
    H_\infty(X_1^k|X_{k+1}^{k+I_k})\lesssim
    \log I_k\lesssim
    \log R^{(2)}_k
    \text{ a.s.}
    ,
  \end{align}
  where the left-most inequality $\lesssim$ becomes the equivalence
  $\sim$ for a finite alphabet.  
\end{theorem}
\noindent
\emph{Remark:} As for the lower bound (\ref{MainThree}), we may
substitute
\begin{align*}
  I_k\to\min\klam{i:\frac{\gamma}{\gamma-1}\log i\ge
  H_\gamma(X_1^k|X_{k+1}^{k+i})}
\end{align*}
for any $\gamma>1$.  Context length $I_k$ can be much shorter than the
pessimistic context length $D^k$ for a $D$-ary alphabet that was
proved sufficient in \cite[Theorem 8]{Debowski18b}.

The organization of this article is as follows. In Section
\ref{secRecurrence}, we recall the known sandwich bound for the
recurrence time. In Section \ref{secRepetition}, we prove a new upper
bound for the repetition time under short memory. 
In Section \ref{secYetAnother}, we discuss a lower bound for the
repetition time under plain stationarity.  Concluding, in Section
\ref{secFinal}, we specialize these results to the Hilberg exponents.

\section{Sandwiching the recurrence time}
\label{secRecurrence}

There is a celebrated lemma by Kac, which states that the conditional
expectation of the recurrence time equals the inverse probability of
the initial string.
\begin{lemma}[\cite{Kac47}]
  \label{theoKac}
  For a stationary ergodic process $(X_i)_{i\in\mathbb{Z}}$ over a
  countable alphabet,
  \begin{align}
    \mean\okra{R^{(1)}_k\middle|X_1^k}
    =
    \frac{1}{P(X_1^k)}.
  \end{align}
\end{lemma}
There is also a converse result by Kontoyiannis \cite{Kontoyiannis98},
which was strengthened by D\k{e}bowski \cite{Debowski18b} from
conditioning on $X_{k+1}^\infty$ to conditioning on
$X_{k+1}^{D^k}$. Here we shorten the conditioning to $X_{k+1}^{N}$,
where $N<D^k$.
\begin{lemma}[cf.\ \cite{Kontoyiannis98}, \mbox{\cite[Lemma 2]{Debowski18b}}]
  \label{theoKontoyiannis}
  For a process $(X_i)_{i\in\mathbb{Z}}$ over a $D$-ary alphabet and
  an integer constant $N$ where $1\le N\le D^k$,
  \begin{align}
    \label{Kontoyiannis}
    \mean\okra{\frac{1}{\min\klam{R^{(1)}_k,N}
    P(X_1^k|X_{k+1}^N)}\middle|X_{k+1}^N}
    \le
    Z(N,D,k):=1+\log N+\frac{D^k-N}{N}.
  \end{align}
\end{lemma}
\begin{proof}
  Denote the left-hand side of (\ref{Kontoyiannis}) as $Y$.  Define
  the waiting time for string $x_1^k$ as
   $R^{(1)}(x_1^k):=\inf\klam{\infty}\cup\klam{i\ge 1: X_{i+1}^{i+k}=x_1^k}$. 
  We observe that for each $i\ge 1$ there is at most one string
  $x_1^k$ such that $R^{(1)}(x_1^k)=i$. Hence, we have a uniform bound
  \begin{align}
    Y
    &=
      \sum_{x_1^k}\frac{1}{\min\klam{R^{(1)}(x_1^k),N}}
      \le
      \sum_{i=1}^{D^k}\frac{1}{\min\klam{i,N}}
      \le
      1+\int_{1}^{D^k}\frac{dx}{\min\klam{x,N}}
      =
      Z(N,D,k).
  \end{align}
\end{proof}

The following sandwich bound essentially appears in the simple proof
of the WZ/OW theorem in \cite[Theorem 9.12]{Debowski21}, which was
originally proved using other techniques, cf.\
\cite{WynerZiv89,OrnsteinWeiss93}. We recall our simpler derivation
for analogy with Section \ref{secRepetition}. Applying Lemma
\ref{theoKontoyiannis}, we slightly shorten the length of conditioning
compared to the results of \cite{Debowski18b}.
\begin{proposition}
  \label{theoSandwichRecurrence}
  Consider a stationary ergodic process $(X_i)_{i\in\mathbb{Z}}$ over
  a $D$-ary alphabet. Define $N_k:=\max\klam{N: N\log N\le D^k}$.
  Assume that $\sum_{k=1}^\infty \rho_k<\infty$ for some $\rho_k>0$.
  Almost surely, for sufficiently large $k$ we have
\begin{align}
  \label{SandwichRecurrence}
  -\log P(X_1^k|X_{k+1}^{N_k})
  +\log \rho_k-\log k
  <
  \log R^{(1)}_k
  <
  -\log P(X_1^k)
  -\log \rho_k
  .
\end{align}
\end{proposition}
\begin{proof}
  By applying the Markov inequality to Lemmas \ref{theoKac} and
  \ref{theoKontoyiannis}, we obtain
\begin{align}
  \label{ProbRecurrence}
  P\okra{R^{(1)}_k\ge \frac{C}{P(X_1^k)}}
  &\le C^{-1}
  ,
  &
  P\okra{R^{(1)}_k\le \frac{C}{P(X_1^k|X_{k+1}^{N})}}
  &\le C Z(N,D,k)
  .
\end{align}
Setting $C=\rho_k^{-1}$ in the left bound of (\ref{ProbRecurrence})
and $C=k^{-1}\rho_k$ with $N=N_k$ in the right bound we obtain
(\ref{SandwichRecurrence}) by the Borel-Cantelli lemma.
\end{proof}

The WZ/OW theorem (\ref{WZOW}) follows by Proposition
\ref{theoSandwichRecurrence} combined with the
Shannon-McMillan-Breiman (SMB) theorem \cite{Shannon48, Breiman57},
which asserts
\begin{align}
  \label{SMB}
  \lim_{k\to\infty}
  \frac{1}{k}\okra{-\log P(X_1^k)}=
  \lim_{k\to\infty}
  \frac{1}{k}\okra{-\log P(X_1^k|X_{k+1}^\infty)}=
  h_1 \text{ a.s.}
\end{align}

\section{Upper bound for the repetition time}
\label{secRepetition}

There is a result by Chen Moy that generalizes the Kac lemma to
successive recurrence times.  Let $T$ be a stationary automorphism of
the probability space.  Consider an event $B$ such that
$P(B)>0$. Define the random set
\begin{align}
  \mathcal{C}^B(\omega):=\klam{j\in\mathbb{Z}: T^{-j}\omega\in B}
\end{align}
If $P(B)>0$ then set $\mathcal{C}^B$ contains infinitely many negative
and positive integers almost surely, by the Poincar\'e recurrence
theorem. Then we can arrange the elements of set $\mathcal{C}^B$ into
the unique sorted sequence $(C^B_m)_{m\in\mathbb{Z}}$ defined by
conditions: $\bigcup_{m\in\mathbb{Z}}\klam{C^B_m}=\mathcal{C}^B$,
$C^B_m>C^B_{m-1}$ for $m\in\mathbb{Z}$, $C^B_1\ge 1$, and
$C^B_0\le 0$.  The successive recurrence times are 
\begin{align}
  \label{Successive}
  W^B_m:=C^B_m-C^B_{m-1}.
\end{align}
Observe that $(C^B_0=0)=B$. The result by Chen Moy is as follows.
\begin{lemma}[\cite{ChenMoy59}]
  \label{theoChenMoy}
  The process $(W^B_m)_{m\in\mathbb{N}}$ is stationary with respect to
  the conditional measure $P(\cdot|B)$. Moreover, if the automorphism
  $T$ is ergodic then for $m\in\mathbb{N}$, we have
  \begin{align}
    \mean\okra{W^B_m\middle|B}
    =
    \frac{1}{P(B)}.
  \end{align}
\end{lemma}

Notice that the repetition time can be expressed as
\begin{align}
  R^{(2)}_k=\min_{x_1^k} C^{(X_1^k=x_1^k)}_2.
\end{align}
Thus, for the analysis of repetition times, we would like to have a
similar bound as Lemma \ref{theoChenMoy} but for the unconditional
expectation. Unfortunately, due to the phenomenon of burstiness, Chen
Moy's result cannot be strengthened to the unconditional expectation.
\begin{example}[periodic process with burstiness]
  Consider a periodic process $(X_i)_{i\in\mathbb{Z}}$ in which the
  pattern of $n_1$ ones followed by $n_0$ zeros repeats
  cyclically. Let $B=(X_0=1)$. Then $P(C^B_0=0)=P(B)=n_1/(n_0+n_1)$
  and $P(C^B_0=-i)=1/(n_0+n_1)$ for $i=1,\ldots,n_0$. For $C_0<0$, we
  have $W^B_{kn_1+l}=1+n_0\boole{l=1}$ for
  $k\in\mathbb{N}\cup\klam{0}$ and $l=1,\ldots,n_1$. Hence
\begin{align}
  \mean W^B_{kn_1+l}&=1+\frac{n_0+n_0^2\boole{l=1}}{n_0+n_1}.
\end{align}
For $n_1=1$, we have $\mean W^B_m=1+n_0=1/P(B)$. However, process
$(W^B_m)_{m\in\mathbb{N}}$ is non-stationary for other values of
$n_1$, which correspond to a burst of $1$'s followed a single lull
that can be arbitrarily longer than $1/P(X_0=1)$.
\end{example}
\noindent
Burstiness can be particularly acute in statistical models of natural
language \cite{Katz96,AltmannPierrehumbertMotter09,Debowski15f}.

Let us try to work around this problem.  Dualities
$R^{(\gamma)}_k>n\iff L^{(\gamma)}_{n+k}<k$ are an instance of a more
general pattern that we propose to call the time-count duality.
\begin{theorem}[time-count duality, cf.\ \cite{WynerZiv89}]
  \label{theoTimeCount}
  Suppose that a function
  $F:\mathbb{N}\rightarrow\mathbb{N}\cup\klam{0}$ (called the count)
  is non-decreasing. We define the dual function
  $C:\mathbb{N}\rightarrow\mathbb{N}\cup\klam{\infty}$ (called the
  time) as
    $C(m):=\inf\klam{\infty}\cup\klam{n\ge 1: F(n)\ge m}$.
  We claim that:
  \begin{enumerate}
  \item $C(m)>n\iff F(n)<m$;    
  \item $F(n)=\max\klam{0}\cup\klam{m\ge 1: C(m)\le n}$;    
  \item The function $C$ is non-decreasing.
  \end{enumerate}
\end{theorem}
\noindent
\emph{Remark:} Since $C$ is non-decreasing, we may apply the dual
construction to $C$ rather than to $F$. But then we obtain a shifted
function $\inf\klam{\infty}\cup\klam{m\ge 1: C(m)\ge n}=F(n-1)+1$.
\begin{proof}
  Since $F$ is non-decreasing, condition $F(n)<m$ implies $C(m)>n$.
  On the other hand, by definition, if $F(n)\ge m$ then $C(m)\le n$.
  Thus we have equivalence $C(m)>n\iff F(n)<m$.  In view of this
  equivalence, we may write
  \begin{align}
    F(n)&=\max\klam{m\ge 0: F(n)\ge m}
    =\max\klam{0}\cup\klam{m\ge 1: C(m)\le n}
    .
  \end{align}
  Suppose that $m_1\le m_2$. We have
  $\klam{n\ge 1: F(n)\ge m_1}\supset \klam{n\ge 1: F(n)\ge m_2}$.
  Since $F$ is non-decreasing, both sets are intervals. Thus
  $C(m_1)\le C(m_2)$.
\end{proof}

We observe that the time-count duality carries over to some extreme
values.
\begin{theorem}
  \label{theoExtreme}
  Let $B$ be a certain set of indices.  Suppose that functions
  $F^B$ and $C^B$ are dual for all $B\in\beta$. Then the functions
  $F:=\sup_{B\in\beta} F^B$ and $C:=\inf_{B\in\beta} C^B$ are also
  dual.
\end{theorem}
\begin{proof}
  The claim follows by equivalence
  \begin{align}
    C(m)>n
    &\iff\forall_{B\in\beta}\,C^B(m)>n
    \iff
      \forall_{B\in\beta}\,F^B(n)<m
      \iff F(n)<m.
  \end{align}
  Hence 
    $C(m)=\inf\klam{\infty}\cup\klam{n\ge 1: C(m)\le n}
    =\inf\klam{\infty}\cup\klam{n\ge 1: F(n)\ge m}$.
\end{proof}

In the following, we denote
the number of visits to $B$ by time $n\in\mathbb{N}$ as
\begin{align}
  F^B_n(\omega):=\sum_{j=1}^n \boole{T^{-j}\omega\in B}.
\end{align}
We have the time-count duality for times
$C^B_m=\inf\klam{\infty}\cup\klam{n\ge 1: F^B_n\ge m}$ and counts
$F^B_n=\max\klam{0}\cup\klam{m\ge 1: C^B_m\le n}$.  Thus we have
equivalence $C^B_m>n\iff F^B_n<m$ by Theorem \ref{theoTimeCount}.
What is useful, counts $F^B_n$ are easier to bound than times $C^B_m$.

Notice that we have expectation $\mean F^B_n=nP(B)$ if the process is
stationary. As a result, despite the phenomenon of burstiness, we can
supply a simple lower bound for the expectation of times $C^B_m$---by
their duality to counts $F^B_n$.
\begin{lemma}
  \label{theoUnconditional}
  If the automorphism $T$ is stationary then for $m\in\mathbb{N}$, we
  have
  \begin{align}
    \mean C^B_m
    \ge
    \frac{m}{2P(B)}.
  \end{align}
\end{lemma}
\begin{proof}
  In view of equivalence $C^B_m>n\iff F^B_n<m$, we have
  $P(C^B_m>n)=P(F^B_n<m)$.  Hence, we may bound the expectation of
  $C^B_m$, using the Markov inequality as follows
  \begin{align}
    \mean C^B_m
    &=
    \int_0^\infty P(C^B_m>\floor{n}) dn
    =
    \int_0^\infty P(F^B_{\floor{n}}<m) dn
      \nonumber\\
    &=
      \int_0^\infty \kwad{1-P(F^B_{\floor{n}}\ge m)} dn
      \ge
      \int_0^{m/P(B)} \kwad{1-\frac{nP(B)}{m}} dn
      =
      \frac{m}{2P(B)}.
  \end{align}
\end{proof}

In the following, we will bound the expectation of random variable
$\min_{B\in\beta} C^B_m$, where $\beta$ is a finite partition.  This
will be achieved under a uniform bound
$\var F^B_n\le\gamma \mean F^B_n$ for a $\gamma<\infty$.  Since random
variable $\min_{B\in\beta} C^B_m$ is surely bounded for a finite
partition $\beta$, in the following proof, we can circumvent a problem
with insummability of the harmonic series and we obtain a bound that
is not void.

\begin{lemma}
  \label{theoShortMemory}
  Assume a stationary automorphism $T$ and a finite partition $\beta$.
  Suppose that
  \begin{align}
    \gamma:=
    \sup_{n\in\mathbb{N}}\max_{B\in\beta}\frac{\var F^B_n}{\mean F^B_n}<\infty.
  \end{align}
  Then for $m>1$, we have
  \begin{align}
    \mean \min_{B\in\beta} C^B_m\le
    \frac{1}{\max_{B\in\beta} P(B)}\okra{
    2m+\gamma\log\frac{(m-1)\card\beta+\gamma}{\gamma}
    }+1
    .
  \end{align}
\end{lemma}
\begin{proof}
  As a consequence of the Paley-Zygmund inequality for a random
  variable $Y\ge 0$, we obtain
  \begin{align}
    P(Y\le y)\le\frac{\var Y}{\var Y+(\mean Y-y)_+^2}.
  \end{align}
  Let $N\in\mathbb{N}$ be a constant. Then the duality yields
  \begin{align}
    \mean \min\klam{C^B_m,N}
    &=
    \int_0^{N} P(C^B_m>\floor{n}) dn
    =
      \int_0^{N} P(F^B_{\floor{n}}<m) dn
      \nonumber\\
    &\le \int_0^{N} \frac{\var F^B_{\floor{n}}
      dn}{\var F^B_{\floor{n}}+(\mean F^B_{\floor{n}}-m)_+^2}
    \le \int_0^{N} \frac{\gamma \floor{n}P(B)
      dn}{\gamma \floor{n}P(B)+(\floor{n}P(B)-m)_+^2}
      \nonumber\\
    &\le \int_0^{N}
      \frac{\gamma dn}{((n-1)P(B)-2m)_++\gamma}
      \le \frac{1}{P(B)}\okra{2m+
      \gamma\log\frac{N-1+\gamma}{\gamma}}+1
      .
  \end{align}
  On the other hand, almost surely, we have
  $\min_{B\in\beta} C^B_m\le N:=(m-1)\card\beta+1$
  since in $(m-1)\card\beta+1$ consecutive draws from set $\beta$
  there must be an element $B\in\beta$ that is repeated $m$ times.
  Consequently, we may write
  \begin{align}
    \mean \min_{B\in\beta} C^B_m
    &\le \mean \min\klam{\min_{B\in\beta} C^B_m,N}
      \le \min_{B\in\beta} \mean \min\klam{C^B_m, N}
      .
  \end{align}
  Hence we obtain the claim.
\end{proof}

Consequently, we can derive an upper bound for the repetition time.
\begin{proposition}
  \label{theoSandwichRepetitionOneOne}
  Consider a stationary process $(X_i)_{i\in\mathbb{Z}}$ over a
  $D$-ary alphabet.  Assume that $\sum_{k=1}^\infty \rho_k<\infty$ for
  some $\rho_k>0$.  Almost surely, for sufficiently large $k$, we have
  \begin{align}
    \label{SandwichRepetitionOneOne}
    \log R^{(2)}_k<H_\infty(X_1^k)+\log
    \okra{4+\gamma_k\log\frac{D^k+\gamma_k}{\gamma_k}}
    -\log\rho_k
    ,
  \end{align}
  where coefficients $\gamma_k$ are defined in (\ref{GammaDelta}).
\end{proposition}
\begin{proof}
  In view of the Markov inequality and Lemma \ref{theoShortMemory},
  we may bound
  \begin{align}
    P(R^{(2)}_k\ge n)
    &\le \frac{\mean R^{(2)}_k}{n}
      \le \frac{\mean \min_{x_1^k} C^{(X_1^k=x_1^k)}_2}{n}
       \le \frac{e^{H_\infty(X_1^k)}
      \okra{4+\gamma_k\log\frac{D^k+\gamma_k}{\gamma_k}}+1}{n}.
  \end{align}
  Define
  $n_k:=e^{H_\infty(X_1^k)}
  \okra{4+\gamma_k\log\frac{D^k+\gamma_k}{\gamma_k}}\rho_k^{-1}$. Then
  we have $\sum_{k=1}^\infty P(R^{(2)}_k\ge n_k)<\infty$. Hence
  inequality (\ref{SandwichRepetitionOneOne}) is obtained by the
  Borel-Cantelli lemma.
\end{proof}

\section{Lower bound for the repetition time}
\label{secYetAnother}

The following proposition refines \cite[Theorem 9.25]{Debowski21} by
using the weighted conditional entropy rather than the min-entropy
conditioned on the infinite future.
\begin{proposition}
  \label{theoSandwichRepetitionTwo}
  Consider a stationary process $(X_i)_{i\in\mathbb{Z}}$ over a
  countable alphabet. Assume that $\sum_{k=1}^\infty \rho_k<\infty$
  for some $\rho_k>0$. Define the weighted conditional entropy
  \begin{align}
    H_-(X_1^k|X_{k+1}^\infty):=-\log \okra{
    \sum_{i=1}^\infty\frac{e^{-H_\infty(X_1^k|X_{k+1}^{k+i})}}{i(i+1)}}.
  \end{align}
  Almost surely, for sufficiently large $k$, we have
  \begin{align}
    \label{SandwichRepetitionTwo}
    \log R^{(2)}_k> \frac{1}{3}H_-(X_1^k|X_{k+1}^\infty)
    +\frac{1}{3}\log\rho_k
    .
  \end{align}
\end{proposition}
\noindent
\emph{Remark:} Since
$H_-(X_1^k|X_{k+1}^\infty)\ge H_\infty(X_1^k|X_{k+1}^\infty)$,
(\ref{SandwichRepetitionTwo}) sharpens \cite[Theorem
9.25]{Debowski21}.
\begin{proof}
  For $i<j$, we obtain 
\begin{align}
  P(X_{i+1}^{i+k}=X_{j+1}^{j+k})
  &=
    \mean
    P(X_{i+1}^{i+k}=X_{j+1}^{j+k}|X_{i+k+1}^{j+k})
    \nonumber\\
  &\le
    \mean \max_{x_1^k} P(X_{i+1}^{i+k}=x_1^k|X_{i+k+1}^{j+k})
    =
    e^{-H_\infty(X_1^k|X_{k+1}^{k+j-i})}
    .
\end{align}
Observe that $3-2x\le x^{-2}$, so $n-x\le 4n^3x^{-2}/27$.
Hence, by the union bound,
\begin{align}
  P(R^{(2)}_k\le n)
  &= P\okra{X_{i+1}^{i+k}=X_{j+1}^{j+k} \text{ for some }
    0\le i<j\le n}
    \le \sum_{0\le i<j\le n}  e^{-H_\infty(X_1^k|X_{k+1}^{k+j-i})}
    \nonumber\\
  &= \sum_{i=1}^{n-1}
    \frac{n-i}{e^{H_\infty(X_1^k|X_{k+1}^{k+i})}}
    \le \frac{8n^3}{27}
    \sum_{i=1}^\infty\frac{e^{-H_\infty(X_1^k|X_{k+1}^{k+i})}}{i(i+1)}
    \le \frac{8n^3}{27} e^{-H_-(X_1^k|X_{k+1}^\infty)}.
\end{align}
Put $n_k:=\sqrt[3]{\rho_k e^{H_-(X_1^k|X_{k+1}^\infty)}}$. Then
$\sum_{k=1}^\infty P(R^{(2)}_k\le n_k)<\infty$ and
(\ref{SandwichRepetitionTwo}) follows by the Borel-Cantelli lemma.
\end{proof}

The goal of the second proposition in this section is to show that the
bound in terms of the weighted conditional entropy is more precise
than the bound proved in \cite[Theorem 8]{Debowski18b} since the
length of conditioning can be shortened to $I_k$ rather than $D^k$.
\begin{proposition}
  \label{theoSandwichContext}
  Consider a stationary process $(X_i)_{i\in\mathbb{Z}}$ over a
  countable alphabet.  For
  $I_k:=\min\klam{i:\log i\ge H_\infty(X_1^k|X_{k+1}^{k+i})}$, we have
\begin{align}
  \label{SandwichContextOne}
  \log I_k-\log 2&\le H_-(X_1^k|X_{k+1}^\infty)\le
  3\log I_k+\frac{1}{I_k}
  ,
  \\
  \label{SandwichContextThree}
  \log (I_k-1)-H_0(X_i)&\le H_\infty(X_1^k|X_{k+1}^{k+I_k})\le \log I_k
                              .
\end{align}
\end{proposition}
\noindent
\emph{Remark:} We note that $I_k\le D^k$ and
$H_\infty(X_1^k|X_{k+1}^{k+I_k})\ge H_\infty(X_1^k|X_{k+1}^{k+D^k})$
for a $D$-ary alphabet. Hence Propositions
\ref{theoSandwichRepetitionTwo} and \ref{theoSandwichContext} combined
are stronger than \cite[Theorem 8]{Debowski18b}.
\begin{proof}
First, we derive (\ref{SandwichContextOne}) from
\begin{align}
  H_-(X_1^k|X_{k+1}^\infty)
  &\le
    H_\infty(X_1^k|X_{k+1}^{k+I_k})+\log I_k+\log (I_k+1)\le
    3\log I_k+\frac{1}{I_k},
  \\
  H_-(X_1^k|X_{k+1}^\infty)
  &\ge
    -\log\okra{\okra{1-\frac{1}{I_k}}e^{-H_\infty(X_1^k|X_{k+1}^{k+I_k-1})}
    +\frac{1}{I_k}}
    \nonumber\\
  &\ge -\log\okra{\okra{1-\frac{1}{I_k}}\frac{1}{I_k-1}+\frac{1}{I_k}}
    =\log I_k-\log 2.
\end{align}
Moreover, for a stationary process,
$0 \le
H_\infty(X_1^k|X_{k+1}^{k+I_k-1})-H_\infty(X_1^k|X_{k+1}^{k+I_k})\le
H_0(X_i)$
since by the chain rule (\ref{Chain}), we have
  $H_\infty(Y,X|Z,V)\le H_\infty(Y,X|Z)
  \le H_0(X)+H_\infty(U,Y|X,Z)$. 
Hence we obtain (\ref{SandwichContextThree}).
\end{proof}

\section{Final touch}
\label{secFinal}

Let us recall some properties of Hilberg exponents.  The probabilistic
results of \cite[Section 8.1]{Debowski21} can be summarized as
follows.  Let $T$ be an automorphism of the probability space.  A
sequence of random variables $(J_k)_{k\in\mathbb{N}}$ is called
$T$-increasing if
$J_{k+1}\ge J_k\circ T,J_k\ge 0$.
For example, for the shift automorphism $X_i\circ T=X_{i+1}$,
sequences $(L^{(\gamma)}_n)_{n\in\mathbb{N}}$,
$(R^{(\gamma)}_k)_{k\in\mathbb{N}}$,
$(-\log P(X_1^k))_{k\in\mathbb{N}}$, and
$(-\log P(X_1^k|X_{k+1}^\infty))_{k\in\mathbb{N}}$ are
$T$-increasing. Subsequently, let
$\median X:=\sup \klam{r: P(X<r)\le \frac{1}{2}}$ denote the median of
a random variable $X$.
\begin{lemma}
  \label{theoHilberg}
  For a stationary ergodic automorphism $T$ and a $T$-increasing
  sequence $(J_k)_{k\in\mathbb{N}}$, we have
  \begin{align}
    \hilberg_{k\rightarrow\infty} \median J_k
    \le \hilberg_{k\rightarrow\infty} J_k 
    \le \hilberg_{k\rightarrow\infty} \mean J_k
    \text{ a.s.}, 
  \end{align}
  where random variable $\hilberg_{k\rightarrow\infty} J_k$ is almost
  surely constant. Moreover, the inequalities become equalities if
  \begin{align}
    \limsup_{k\rightarrow\infty} \frac{\sqrt{\var J_k}}{\mean J_k}<1.
  \end{align}
\end{lemma}
\begin{proof}
  All claims but the last one follow by \cite[Theorems 8.3--8.5]{Debowski21}.
  The last claim is a slight strengthening of \cite[Theorem
  8.6]{Debowski21} and in a similar fashion it follows by the
  median-mean inequality $\abs{\median X-\mean X}\le \sqrt{\var X}$,
  see e.g.\ \cite[Theorem 3.17]{Debowski21}.
\end{proof}

Thus Theorem \ref{theoMainOne} follows by Proposition
\ref{theoSandwichRecurrence} and Lemma \ref{theoHilberg}, Theorem
\ref{theoMainTwo} follows by Propositions
\ref{theoSandwichRepetitionOneOne}, and Theorem \ref{theoMainThree}
follows by Propositions
\ref{theoSandwichRepetitionTwo}--\ref{theoSandwichContext}.

\bibliographystyle{abbrvnat}

\bibliography{0-journals-abbrv,0-publishers-abbrv,ai,ql,mine,tcs,books,nlp}

\section*{Acknowledgment}

I thank the reviewers for stimulating reviews that encouraged me to
inspect various short memory conditions from the stance of summability
of different mixing rates.

\end{document}